\documentclass{article}
\usepackage{ijcai24}

\usepackage{times}
\usepackage{soul}
\usepackage{url}
\usepackage[hidelinks]{hyperref}
\usepackage[utf8]{inputenc}
\usepackage[small]{caption}
\usepackage{graphicx}
\usepackage{amsmath}
\usepackage{amsthm}
\usepackage{booktabs}

\newtheorem{theorem}{Theorem}

\def\arxiv{}

\ifdefined\arxiv
    \usepackage{sgame}
    \usepackage[style=authoryear,maxbibnames=99,dashed=false]{biblatex}
    \bibliography{references}
    \providecommand{\citet}[2][]{\ifx&#1&\textcite{#2}\else\textcite[#1]{#2}\fi}
    \renewcommand{\cite}[1]{\parencite{#1}}
\else
    \providecommand{\citet}[2][]{\citeauthor{#2}~\ifx&#1&\shortcite{#2}\else\shortcite[#1]{#2}\fi}
\fi

\title{Equilibria in Two-Stage Facility Location with Atomic Clients\ifdefined\arxiv\else\footnote{A full version is available at \url{http://arxiv.org/abs/2403.03114}.}\fi}

\author{
Simon Krogmann$^1$
\and
Pascal Lenzner$^1$\and
Alexander Skopalik$^2$\and
Marc Uetz$^2$\And
Marnix C. Vos$^2$\\
\affiliations
$^1$Hasso Plattner Institute, University of Potsdam\\
$^2$Mathematics of Operations Research, University of Twente\\
\emails
\{simon.krogmann, pascal.lenzner\}@hpi.de,
\{a.skopalik, m.uetz\}@utwente.nl,
m.c.vos@student.utwente.nl
}

\usepackage{tikz}
\usepackage{thm-restate}
\usepackage[ruled,vlined,linesnumbered]{algorithm2e}
\SetAlgoNlRelativeSize{-1.9}
\usepackage{todonotes}
\usepackage[nameinlink]{cleveref}
\usepackage{amssymb}
\usepackage{enumitem}
\newtheorem{observation}{Observation}
\newtheorem{definition}{Definition}
\newtheorem{corollary}{Corollary}
\newtheorem{lemma}{Lemma}

\crefname{step}{step}{steps}
\newcommand{\flg}{\texorpdfstring{$2$}{2}-FLG}
\newcommand{\s}{\mathbf{s}}
\newcommand{\lsort}{\ell_\mathrm{sort}}
\newcommand{\lpi}{\ell^\pi}
\newcommand{\C}{\mathcal{C}}
\newcommand{\MNS}{\mathrm{MNS}}

\newcommand{\oneProfile}{client profile}
\newcommand{\fullProfile}{full client profile}
\newcommand{\oneEq}{client equilibrium}
\newcommand{\oneEqa}{client equilibria}
\newcommand{\fullEq}{full client equilibrium}
\newcommand{\fullEqa}{full client equilibria}
\ifdefined\arxiv
\newcommand{\suppmaterial}{appendix}
\newcommand{\para}[1]{\paragraph{#1}}
\else
\newcommand{\suppmaterial}{full version~\cite{arxiv}}
\newcommand{\para}[1]{\textbf{#1}}
\fi

\newcommand*{\trheight}{1.03923}
\newcommand*{\trlength}{1.2}
\newcommand*{\noderadius}{6.5pt}
\usetikzlibrary{arrows.meta, quotes, calc, fit}
\tikzset{
node distance={\trlength cm}, vert/.style = {draw, circle, inner sep = 0 pt, minimum size = 2 * \noderadius}, every path/.style = {-Latex}, every label/.append style={rectangle, font = {\footnotesize}}, fac/.style = {circle,fill,inner sep=1.5pt}, every node/.style = {font = {\footnotesize}}, every edge quotes/.style = {auto, sloped, font = {\footnotesize}, inner sep = 0.5pt}
}

\definecolor{cred}{HTML}{D81B60}
\definecolor{cblue}{HTML}{1E88E5}
\definecolor{cyellow}{HTML}{D09C00}
\definecolor{cgreen}{HTML}{5B8600}

\begin{document}

\maketitle

\begin{abstract}
We consider competitive facility location as a two-stage multi-agent system with two types of clients.
For a given host graph with weighted clients on the vertices, first facility agents strategically select vertices for opening their facilities.
Then, the clients strategically select which of the opened facilities in their neighborhood to patronize.
Facilities want to attract as much client weight as possible, clients want to minimize congestion on the chosen facility.

All recently studied versions of this model assume that clients can split their weight strategically.
We consider clients with unsplittable weights but allow mixed strategies.
So clients may randomize over which facility to patronize.
Besides modeling a natural client behavior, this subtle change yields drastic changes, e.g., for a given facility placement, qualitatively different client equilibria are possible.

As our main result, we show that pure subgame perfect equilibria always exist if all client weights are identical.
For this, we use a novel potential function argument, employing a hierarchical classification of the clients and sophisticated rounding in each step.
In contrast, for non-identical clients, we show that deciding the existence of even approximately stable states is computationally intractable.
On the positive side, we give a tight bound of~$2$ on the price of anarchy which implies high social welfare of equilibria, if they exist.
\end{abstract}

\section{Introduction}
In classical facility location~\cite{cornuejols1983uncapicitated,korte2011combinatorial}, a central authority places facilities in some underlying space to serve a set of clients optimally.
While this might be realistic for applications like the construction of public hospitals, in many other domains facilities are placed by selfish agents competing for clients' attention.
For example, supermarkets, pubs, or fast food restaurants aim to maximize profits by trying to attract nearby clients.

Beginning with the seminal model by \citet{hotelling} and \citet{downs} in which facilities compete for clients on a line and clients always patronize their closest facility, many variants of competitive facility location models have been investigated, see~\cite{eiselt1993competitive,revelle2005location,brenner2010location} for an overview.
However, a notable feature of most competitive facility location models is that clients do not influence each other, i.e., their behavior only depends on the location of the facilities but not on the behavior of the other clients.
Thus, facility agents face strategic decisions but the clients do not.

In contrast to these one-sided variants, recently two-sided versions featuring a client subgame have been introduced~\cite{ijcai-21,aaai-23}, making it a sequential two-stage game.
In the first stage, the facility agents simultaneously select a location for their facilities, while in the second stage, the clients simultaneously select which facilities to patronize.
In these models, the clients do interact with each other, because the client weight at a facility influences the waiting time
%more clients increase the waiting time at a facility
which is the cost that the clients aim to minimize.
%Clients then aim to visit facilities within a given range that have minimal congestion, so as to minimize their own waiting times.
In the previously studied models, clients are non-atomic, so they may distribute their weight among multiple facilities.
%Modeling, e.g., shipping companies, that select ports to deliver their goods, they
The induced client subgame is a non-atomic congestion game with an essentially unique client Nash equilibrium for any given facility placement.
This yields a two-stage facility location game where the facilities can perfectly predict the behavior of the clients.
%, since they only have to compute the (unique) client Nash equilibrium.
As a consequence, subgame perfect equilibria, the suitable solution concept for sequential games, are comparatively easy to reason about.

We take this line of work to the next level by considering \emph{atomic clients}, which cannot split their weights among multiple facilities.
However, clients are allowed to play mixed strategies, meaning that they can randomize over which facility to patronize.
Given a facility placement, the induced client subgame is then still a well-studied game, namely a singleton congestion game~\cite{Rosenthal1973,fotakis-congestion}.
Arguably, this setting is more realistic in many scenarios, like supermarket shopping or restaurant visits.
The reason is that clients, conditioned on selecting one specific facility, contribute with their full weight to the congestion of the chosen facility.
For example, a mixed strategy of an atomic client could represent selecting different facilities on different days.

%because an atomic client always experiences her full weight as latency, even when playing a mixed strategy.

Besides capturing a wide range of realistic facility location settings, a consequence of introducing atomic clients is the existence of multiple, potentially different Nash equilibria in the client subgame for a given facility placement.
These different equilibria may also yield different congestion at facilities, so the client behavior is harder to predict for the facility agents.
As a consequence, the two-stage game considered in this paper requires more involved techniques for reasoning about subgame perfect equilibria and their existence.

We develop such techniques to analyze this arguably more complex two-stage competitive facility location setting and we believe that our approach will prove versatile for understanding other multi-stage models as well.
One key technique is a novel potential argument that exploits the two-staged nature of our model and the multiplicity of equilibria in the second stage.
In particular, we use a hierarchical classification of the clients and sophisticated rounding in each step.
This ensures that for any improving strategy change by a facility, we can carefully select a new suitable client equilibrium to sustain this improvement in the potential function value.

\subsection{Model and Preliminaries}

\para{The Basics.} We consider the \emph{two-stage facility location game with atomic clients}, \emph{atomic \flg{}} for short, where a set of $k$ facility agents $F$ and a set of $n$ client agents $V$ interact on a given vertex-weighted directed host graph $H = (V,E,w)$, with weight function $w: V \to \mathbb{Q}^+$ that assigns every vertex a positive rational weight.
In \Cref{sec:spe-unweighted}, we consider all clients $v$ to be unweighted, i.e., $w(v)=1$.
We let the total weight of the vertices in $X \subseteq V$ be $w(X) := \sum_{v\in X} w(v)$.

% TODO journal: maybe make another figure that stands on its own and does not compare to the non-atomic model

The vertices of $H$ correspond to the client agents but at the same time, they also serve as possible locations for the facility agents.
The weight of a vertex can be understood as the purchasing power or expected revenue of the corresponding client agent or location.
Every facility agent $f\in F$ can choose to locate her facility at a vertex $v\in V$, however for facility agent $f$, the feasible vertices may be restricted to $U(f) \subseteq V$.
By $U: F \to 2^V$ we denote the function that maps facility agents to feasible vertices.
Thus, an instance of the {atomic \flg{}} is specified by the triple $(H,U,k)$.
If $U(f) = V$ for every facility agent $f$ then we say that the instance is \emph{unrestricted} and we will omit $U$ in this case.

\para{The Sequential Game.} The atomic \flg{} consists of two stages.
First, each facility agent $f$ selects a location $s_f \in U(f) \subseteq V$ for opening a facility on host graph $H$.
Let $\s = (s_1,\dots,s_k)$ denote the vector of chosen facility locations for some ordering of the facility agents.
Note that $s_f=s_g$ is possible for different facilities $f,g\in F$.
We say that $\s$ is the \emph{facility placement profile (FPP)}.
Moreover, let $S\in V^k$ denote the set of all possible FPPs.
In the second stage, given a FPP,
each client agent $v\in V$ strategically decides which facility to patronize.
A client agent $v$ can only patronize a facility $f$ that is located in her neighborhood $N(v) = \{v\} \cup \{u \mid (v,u) \in E\}$, so if $s_f\in N(v)$.
Let $N_\s(v) = \{f\in F \mid s_f \in N(v)\}$ denote the set of facilities that client $v$ can potentially patronize, for a given FPP $\s$.
Conversely, we define the \emph{attraction range} of a facility $f$ to be $A_\s(f) = \{v \in V \mid f \in N_\s(v)\}$.
We overload this function for a set of facilities $X \subseteq F$, i.e., $A_\s(X) = \bigcup_{f\in X} A_\s(f)$, and $N_\s(X)=\bigcup_{v\in X} N_\s(v)$, for a set of clients $X \subseteq V$.

A main distinguishing feature of this work is that we consider \emph{atomic} clients.
That means that a client $v\in V$, when patronizing a facility $f\in N_\s(v)$, uses her full purchasing power $w(v)$ exclusively on facility $f$.
However, client agents can play mixed strategies, i.e., client agents are allowed to randomize which facilities to visit.

\para{Feasible Client Strategies.}
For a given FPP $\s$, we will denote by $\sigma(\s)$ a \emph{\oneProfile}.
The strategy of a client $v\in V$ is $\sigma(\s)_v$, which is a probability distribution on $N_\s(v)$,
the set of facilities that client $v$ can patronize.
By $\sigma(\s)_{v,f} \in [0,1]$ we denote the probability that client $v$ patronizes facility $f \in N_\s(v)$.
For given $\s$, no facility might be available for some clients $v$, i.e., $N_\s(v) = \varnothing$.
In this case, such clients do not patronize any facility. We therefore define a strategy profile, for given FPP~$\s$, as a function $\sigma(\s): V \to [0,1]^k$ that has to fulfill the feasibility conditions:
For each client $v\in V$,
\begin{itemize}\itemsep0pt
\item[(i)] $\sigma(\s)_{v,f} = 0$, for all $f\notin N_\s(v)$, and
\item [(ii)]
if $N_\s(v)\neq \varnothing$ then $\sum_{f \in N_\s(v)} \sigma(\s)_{v,f} = 1$.
\end{itemize}
Since we have a sequential game with two stages, the strategy of the client agents must be specified for all possible FPPs~$\s$.
Thus, the \emph{\fullProfile{}} is defined by a function $\sigma: S \times V \to [0,1]^k$.
Such a \fullProfile{} $\sigma$ is feasible if $\sigma(\s)$ is feasible for every $\s \in S$.
Moreover, let $\Sigma$ denote the set of all feasible \fullProfile{}s and let $\Sigma_{\s}$ denote the set of all feasible \oneProfile{}s for a given FPP~$\s$.

\para{Objectives.}
An \emph{outcome} of the {atomic \flg{}} on a given host graph $H$ is defined by the tuple $(\s,\sigma)$, where $\s \in S$ and $\sigma \in \Sigma$.
Any outcome entails \emph{expected facility loads}, where the expected facility load of facility $f$ in outcome $(\s,\sigma)$ is
\[
\ell_f(\s,\sigma) = \sum\nolimits_{v\in V}\sigma(\s)_{v,f} \cdot w(v)\text.
\]
Thus, $\ell_f(\s,\sigma)$ captures the expected demand that is realized at facility $f$ under client behavior $\sigma$.
Naturally, we assume that facility agents seek a facility placement to maximize this expected demand.
From the viewpoint of the clients, however, the expected facility loads are seen as congestion, i.e., as waiting time at the given facility.
We assume that clients want to minimize their \emph{expected waiting time} $L_v(\s,\sigma)$, being %defined via the expected facility loads as follows:
\begin{align*}
%L_v(\s,\sigma) &= \sum_{f\in N_\s(v)}\sigma(\s)_{v,f} \cdot \ell_f(\s,\sigma)\\
L_v(\s,\sigma)
%&= \sum_{f\in N_\s(v)}\sigma(\s)_{v,f} \cdot E[\ell_f(\s,\sigma)\mid v \text{ patronizes } f]\\
& = w(v) + \sum\nolimits_{f\in N_\s(v)}\sigma(\s)_{v,f}\cdot \ell_{-v,f}(\s,\sigma)\text,
\end{align*}

where $\ell_{-v,f}(\s,\sigma) = \sum_{u\neq v}\sigma(\s)_{u,f} \cdot w(u)$ is the expected load of facility $f$ contributed by all client agents other than $v$.
We also call $\ell_{-v,f}(\s,\sigma)$ the \emph{$v$-excluded load of facility $f$}.

\begin{figure}[t]
    \centering
    \begin{tikzpicture}
        \coordinate(a) at (6cm +\noderadius,0);
        \coordinate(b) at (6cm + 1.8*\trheight cm +\noderadius,0);
        
        \fill[cblue!40] (a) arc (0:180:\noderadius) -- cycle;
        \fill[cblue!40] (b) arc (0:180:\noderadius) -- cycle;
        \fill[cred!40] (a) arc (0:-180:\noderadius) -- cycle;
        \fill[cred!40] (b) arc (0:-180:\noderadius) -- cycle;
        \draw[ultra thin, -] (a) -- ++(-2*\noderadius,0);
        \draw[ultra thin, -] (b) -- ++(-2*\noderadius,0);
    
        \foreach \i / \name / \filla / \fillb in {0/a/cred!40/cblue!40,3/b/cblue!40/cred!40,6/none/none} {
            \node (\name_c1) [vert,fill=\filla] [label={[inner ysep=0pt, yshift=3pt]above:$3$}] at (\i+0,0) {};
            \node (\name_c2) [vert,fill=\fillb] [label={[inner ysep=0pt, yshift=3pt]above:$1$}] at (\i+1.8*\trheight,0) {};
            \node (\name_f1) [vert] [label={[inner ysep=0pt, xshift=-2pt, yshift=5pt]right:$0$}] at (\i+0.9*\trheight,0.26*\trheight) {};
            \node (\name_f2) [vert] [label={[inner ysep=0pt, xshift=-2pt, yshift=-5pt]right:$0$}] at (\i+0.9*\trheight,-0.26*\trheight) {};
            \node[fac,cblue] at (\name_f1) {};
            \node[fac,cred] at (\name_f2) {};
            \draw (\name_c1) edge (\name_f1) edge (\name_f2);
            \draw (\name_c2) edge (\name_f1) edge (\name_f2);
        }
    \end{tikzpicture}
    \caption{
    Instance with four atomic clients ($=$ vertices) with weights $3,1,0,0$.
    Depicted are the three \oneEqa{} for one given FPP where the two facilities ($=$ colored bullets) each select a $0$-weight node as location.
    The (mixed) client strategies of the clients with weights $3$ and $1$ are shown as pie charts within the client nodes, e.g., on the right, both clients select each facility with probability $\frac12$.
    %On the left, the red client cannot improve through a mixed strategy, as she will experience her full weight on either facility.
    }
    \label{fig:example}
\end{figure}
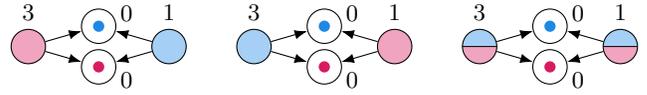

\para{Equilibria.}
For a given FPP $\s$, we say that $\sigma(\s)$ is a \oneEq{}, if no client $v\in V$ has an alternative client strategy $\sigma'(\s)_v$ such that $L_v(\s,(\sigma'(\s)_v, \sigma(\s)_{-v})) < L_v(\s,\sigma(\s))$, where, as usual, $\sigma(\s)_{-v}$ denotes the strategies of all clients except~$v$.
See \Cref{fig:example} for examples.

A state $(\s, \sigma)$ is a subgame perfect equilibrium (SPE) if
\begin{itemize}\itemsep0pt
\item[(1)] $\sigma(\s')$ is a \oneEq{} for \emph{each} FPP $\s'$ and
\item[(2)]for no facility $f$ there exists a location $s'_f$ such that $\ell_f((s'_f, \s_{-f}), \sigma) > \ell_f(\s, \sigma)$.
\end{itemize}
%(1) $\sigma(\s')$ is a \oneEq{} for \emph{each} FPP $\s'$ and (2) for no facility $f$ there exists a location $s'_f$ such that $\ell_f((s'_f, \s_{-f}), \sigma) > \ell_f(\s, \sigma)$.
\noindent For an $\alpha$-approximate SPE, we replace (2) with the following relaxation:
For no facility $f$ there exists a location $s'_f$ such that $\ell_f((s'_f, \s_{-f}), \sigma) > \alpha\cdot\ell_f(\s, \sigma)$, for some $\alpha\ge 1$.

\para{Social Welfare.}
We measure the social welfare of a facility placement profile $\s$ using the \emph{weighted participation rate} $w(\s)=\sum_{v: N_\s(v) \neq \varnothing}{w(v)}$, which is the total client weight covered by at least one facility.
%Note that the weighted participation rate only depends on the facility placement and not on the client strategies.
With this, we define the price of anarchy (PoA)~\cite{poa} as
\[
\text{PoA} =\max_{H,U,k}
\frac{
w(\text{OPT}(H,U,k))}{w(\text{worstSPE}(H,U,k))}\text,
\]
where $\text{OPT}(H,U,k)$ is the FPP with the highest weighted participation rate and $\text{worstSPE}(H,U,k)$ the FPP $\s$ for the SPE $(\s, \sigma)$ with the lowest weighted participation rate.
\ifdefined\arxiv
We also define the price of stability (PoS)~\cite{pos} as
\[
\text{PoS} =\max_{H,U,k}
\frac{
w(\text{OPT}(H,U,k))}{w(\text{bestSPE}(H,U,k))}\text,
\]
where $\text{bestSPE}(H,U,k)$ is the FPP $\s$ for the SPE $(\s, \sigma)$ with the highest weighted participation rate.
\fi
Observe that the PoA \ifdefined\arxiv and the PoS are \else is \fi only well-defined if a SPE exists with a weighted participation rate greater than $0$.

\subsection{Related Work}
For one-sided competitive facility location, the Hotelling-Downs model was also analyzed on graphs~\cite{palvolgyi2011hotelling,Fournier19,FournierS19}, where the clients reside on the edges.
A discrete version of this is the Voronoi game~\cite{vornoi04,voronoi}, with facilities and clients placed on the nodes of an underlying network.
Also, graph-based models with limited attraction range of facilities have been studied~\cite{feldman-hotelling,ShenW17,CohenP19}, where clients patronize facilities only within a certain distance and split their weight equally among them.
Other non-cooperative variants have been investigated, e.g., \cite{Vetta02,CardinalH10,SabanM12}.

To the best of our knowledge, the first model where the clients also face strategic decisions was proposed by \citet{KOHLBERG}.
In this model, the clients' cost function is a linear combination of distance and waiting time.
Kohlberg shows the existence of a SPE for two facilities.
Later, \citet{Peters2018} extended this to four and six facilities, if the cost function is heavily tilted towards waiting time, and \citet{FeldottoLMS19} investigate approximate SPE.

Closest to our work are the models by~\citet{ijcai-21} and \citet{aaai-23}.
They consider non-atomic clients that split their weights among nearby facilities to either minimize (1) their maximum waiting time or (2) their total waiting time.
While for (1) SPE exist, for (2) deciding SPE existence is NP-hard but approximate SPE can be computed efficiently.
In contrast to both variants, we consider clients who cannot split their weight among multiple facilities.
However, clients may use randomized strategies.
This subtle difference is important since it yields different equilibria, i.e., randomizing over atomic strategies is not the same as non-atomic weight splitting, see~\Cref{fig:splittable}.
Moreover, we also allow agent-dependent restrictions for the placement of facilities, which makes our model more general.
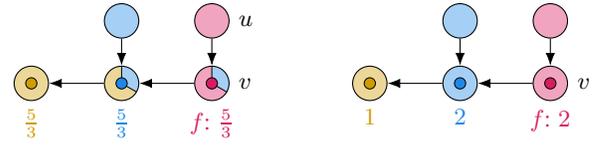
\begin{figure}
    \centering
    \begin{tikzpicture}
        \coordinate(2c) at (\trlength, 0);
        \coordinate(3c) at (2*\trlength, 0);

        % coloring of split nodes
        \fill[cyellow!40] (2c) -- ++(0,\noderadius) arc (90:330:\noderadius) -- cycle;
        \fill[cblue!40] (2c) -- ++(0,\noderadius) arc (90:-30:\noderadius) -- cycle;
        
        \fill[cred!40] (3c) -- ++(0,\noderadius) arc (90:330:\noderadius) -- cycle;
        \fill[cblue!40] (3c) -- ++(0,\noderadius) arc (90:-30:\noderadius) -- cycle;

        % dividers in split nodes
        \draw[ultra thin, -] (2c) -- ++(90:\noderadius);
        \draw[ultra thin, -] (2c) -- ++(330:\noderadius);
        \draw[ultra thin, -] (3c) -- ++(90:\noderadius);
        \draw[ultra thin, -] (3c) -- ++(330:\noderadius);

        % nodes left
        \node (1a) [vert,fill=cyellow!40, label={[cyellow, inner ysep=0pt, yshift=-3pt]below:$\frac53$}] at (0,0) {};
        \node (2a) [vert,label={[cblue, inner ysep=0pt, yshift=-3pt]below:$\frac53$}] at (2c) {};
        \node (3a) [vert,label={[cred, inner ysep=0pt, yshift=-3pt]below:{$f$: $\frac53$}}] [label=right:$v$] at (3c) {};
        \node (4a) [vert,fill=cblue!40] at (1*\trlength,0.8*\trheight){};
        \node (5a) [vert,fill=cred!40] [label=right:$u$] at (2*\trlength,0.8*\trheight){};

        % nodes right
        \node (1b) [vert,fill=cyellow!40, label={[cyellow, inner ysep=0pt, yshift=-3pt]below:$1$}] at (4.5, 0) {};
        \node (2b) [vert,fill=cblue!40, label={[cblue, inner ysep=0pt, yshift=-3pt]below:$2$}, right of=1b]{};
        \node (3b) [vert,fill=cred!40, label={[cred, inner ysep=0pt, yshift=-3pt]below:{$f$: $2$}}, right of=2b] [label=right:$v$] {};
        \node (4b) [vert,fill=cblue!40] at (4.5 + 1*\trlength,0.8*\trheight){};
        \node (5b) [vert,fill=cred!40] at (4.5 + 2*\trlength,0.8*\trheight){};

        % facilities, edges
        \foreach \name in {a,b} {
            \node[fac,cyellow, draw=black, ultra thin] at (1\name) {};
            \node[fac,cblue, draw=black, ultra thin] at (2\name) {};
            \node[fac,cred, draw=black, ultra thin] at (3\name) {};
            \draw (2\name) -- (1\name);
            \draw (3\name) -- (2\name);
            \draw (4\name) -- (2\name);
            \draw (5\name) -- (3\name);
        }
    \end{tikzpicture}
    \caption{Instance with unit weight clients ($=$ vertices) and FPP~$\s$ for three facilities ($=$ colored bullets).
    Left: \oneEq{} with \emph{non-atomic} clients.
    The facilities receive equal load, e.g., red facility $f$ receives the full demand of $u$, and $\frac23$ of the demand of $v$.
    This is also a mixed \oneProfile{} for \emph{atomic} clients, but \emph{not} an equilibrium as client $v$ can improve by increasing her probability of patronizing~$f$.
    Right: A possible \oneEq{} for atomic clients.}
    \label{fig:splittable}
\end{figure}

%All so far mentioned models consider clients that want to minimize their distance to a facility and/or their waiting time.
However, also clients with different types that aim for a certain type-distribution at their selected facility have been considered~\cite{HarderKLS23} and this is closely related to Hedonic Diversity Games~\cite{BredereckEI19}.

Besides the analysis of competitive facility location, many recent works~\cite{KanellopoulosVZ22,Walsh22,DeligkasFV22,MaX0K23} consider a mechanism design point of view.
There, clients submit their locations to a mechanism that selects positions for opening facilities.
See~\cite{mech-design-survey} for an overview.

%Our main object of study are subgame perfect equilibria. This is a standard solution concept in game theory~\cite{peters-textbook-4}. For games with perfect information, subgame perfect equilibria can be computed efficiently via backward induction. However, since in our model all facilities move simultaneously in the first stage%and since different client equilibria exist
%, we do not have a setting with perfect information.

\subsection{Our Contribution}
To the best of our knowledge, the two-stage facility location game with atomic, i.e., unsplittable clients has not been considered before.
It is more complex than the earlier considered two-stage facility location games with splittable clients, as client equilibria are no longer unique, and may even be qualitatively different with respect to their facility loads.
We first prove the existence of subgame perfect equilibria for clients with unit weights, by using a modified best response dynamic.
For clients with arbitrary weights, we show that the existence of subgame perfect equilibria is not guaranteed, and that it is even NP-complete to decide if a $\phi$-approximate subgame perfect equilibrium exists, for the golden ratio $\phi\approx 1.618$.
Using the utilitarian social welfare of the facilities, which equals the weighted participation rate, we also prove a tight bound of $2$ on the price of anarchy, given that an equilibrium exists.

\section{Client Equilibria}
\label{sec:client-eq}
Given a FPP~$\s$, the client subgame is an example of the well-studied \emph{(weighted) singleton congestion game}.
The unweighted version has the finite improvement property and thus an improving response dynamic always converges to a pure Nash equilibrium~\cite{Rosenthal1973}.
For the weighted version, \citet{fotakis-congestion} show that pure Nash equilibria always exist and can be computed efficiently.

We first observe that in a (mixed) \oneEq{} a client $v$ only patronizes certain facilities in her shopping range.

\begin{observation}[Client Patronage Set] \label{obs:excluded-load}
    A \oneProfile{} $\sigma(\s) \in \Sigma_\s$ is a \oneEq{} for a given FPP~$\s$ if and only if for each $v \in V$, the set of her patronized facilities $\{f \in F \mid \sigma(\s)_{v,f} > 0\}$ contains only facilities with minimal $v$-excluded load within her shopping range $N_\s(v)$.
\end{observation}

Thus, we may check in polynomial time if a \oneProfile{} is a \oneEq{}.
Client equilibria are however not necessarily unique and might induce different facility loads.

\begin{observation}[Qualitatively Different Client Equilibria] \label{thm:client-eq-not-unique}
    For a given facility placement profile, the number of client equilibria is not necessarily bounded and different client equilibria may induce different facility loads.
\end{observation}
\begin{proof}
    Consider an unweighted instance with a single node $v$ and $k=2$ facility agents placed on $v$.
    For every $\gamma \in [0,1]$, the \oneProfile{} $\sigma(\s)$ with $\sigma(\s)_v= (\gamma, 1-\gamma)$ is a \oneEq{}.
    Furthermore, the expected facility loads $\gamma$ and $1-\gamma$ differ for every value of $\gamma$.
\end{proof}

\section{Subgame Perfect Equilibria}
\label{sec:spe-unweighted}

We now prove the existence of SPE for the \emph{unweighted} atomic \flg{} via slightly modified improving response dynamics for the facility agents.
We use the sorted vector of facility loads as a potential function, but after each facility move, we also select a new \fullEq{}.
This is necessary because there exist improving facility moves that decrease the value of the potential function.
We prevent these moves by changing to a \fullEq{} that consistently favors the facilities that already receive more load.
We show that this change does not affect the value of the potential function.

\subsection{Rounded Client Profiles}

To determine suitable candidates for strategy profiles that yield SPE, we first partition the set of clients and the set of facilities into \emph{classes}.
To that end, we implicitly employ a process that repeatedly removes those facilities that have the least (average) number of clients in their range.
To achieve this, we define the \emph{minimum neighborhood set}\footnote{We note the difference to the definition of \citet{ijcai-21}, as we ensure uniqueness of the MNS.
Furthermore, the explicit definition of $F_i$ and $V_i$ allows for easier recursion.} as the set of facilities that has the smallest average number of clients in its attraction range for a given subset of clients and facilities.
\begin{definition}[Minimum Neighborhood Set (MNS)]
    For a facility placement profile $\s$, a subset $F^* \subseteq F$ and another subset $V^* \subseteq V$, we denote the \emph{minimum neighborhood set} $\MNS_\s(F^*, V^*)$ as the largest subset of $F^*$ with
    \[
        \MNS_\s(F^*, V^*) \in \arg\min_{T \subseteq F^*} \frac{w(A_\s(T) \cap V^*)}{|T|}\text.
    \]
\end{definition}
We now repeatedly remove a MNS and its associated set of clients from the game and assign both sets to a class.

\begin{definition}[Class Set]
For $i\ge 1$, the class $C_i=(F_i,V_i)$ is inductively defined by

\noindent\resizebox{\linewidth}{!}{
\( \displaystyle
    F_i=\MNS_\s\biggl(F\setminus \bigcup_{j=1}^{i-1} F_j, V\setminus \bigcup_{j=1}^{i-1} V_j\biggr)
    \ \text{and}\ 
    V_i = A_\s(F_i) \setminus \bigcup_{j=1}^{i-1} V_{j}\text.
\)
}
For class $C_i$, denote its average load $\ell(C_i):=\frac{w(V_i)}{|F_i|} $.
For a FPP~$\s$, we call $\C:=\{C_1,C_2,\ldots\}$ the \emph{class set} of $\s$.
For a client $v$, let $C(v)$ be the class of $v$, i.e., $C(v)= C_i$ with $v \in C_i$.
Similarly, $C(f)$ is the class containing facility~$f$.
\end{definition}

An example is shown in \Cref{fig:class-set}.
We note that due to minimum neighborhood sets being of maximal cardinality, the loads of classes are pairwise different and increasing.

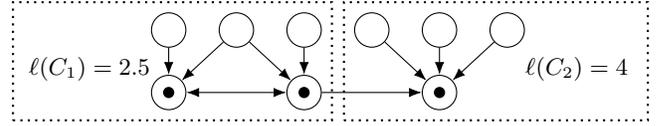
\begin{figure}
    \centering
    \begin{tikzpicture}
        \node[vert] (f1) at (0*1.5*\trlength, 0*\trheight) {};
        \node[vert] (f2) at (1*1.5*\trlength, 0*\trheight) {};
        \node[vert] (f3) at (2*1.5*\trlength, 0*\trheight) {};
        \node[vert] (v1) at (0*1.5*\trlength, 0.8*\trheight) {};
        \node[vert] (v2) at (0.5*1.5*\trlength, 0.8*\trheight) {};
        \node[vert] (v3) at (1*1.5*\trlength, 0.8*\trheight) {};
        \node[vert] (v4) at (1.5*1.5*\trlength, 0.8*\trheight) {};
        \node[vert] (v5) at (2*1.5*\trlength, 0.8*\trheight) {};
        \node[vert] (v6) at (2.5*1.5*\trlength, 0.8*\trheight) {};
        
        \node[fac] at (f1) {};
        \node[fac] at (f2) {};
        \node[fac] at (f3) {};

        \draw[thick,dotted] ($(v3) + (0.37, 0.37)$) rectangle ($(f1) - (0.37 + 1.7, 0.37)$);
        \draw[thick,dotted] ($(v6) + (0.37 + 1.5, 0.37)$) rectangle ($(v4) - (0.37, 0.8*\trheight + 0.37)$);
        \node[draw=none, anchor=east] at ($(v1) - (0.1*\trlength, 0.5*\trheight)$) {$\ell(C_1)=2.5$};
        \node[draw=none, anchor=west] at ($(v6) + (0.1*\trlength, -0.5*\trheight)$) {$\ell(C_2)=4$};

        \draw (f1) edge[Latex-Latex] (f2);
        \draw (f2) edge (f3);
        \draw (v1) edge (f1);
        \draw (v2) edge (f1) edge (f2);
        \draw (v3) edge (f2);
        \draw (v4) edge (f3);
        \draw (v5) edge (f3);
        \draw (v6) edge (f3);
        
    \end{tikzpicture}
    \caption{In this unweighted instance, the class set contains the shown two classes for the given facility placement profile.}
    \label{fig:class-set}
\end{figure}

\begin{corollary}[Uniqueness of Class Set]
\label{cor:class-set-unique}
For an instance of the \flg{} and the FPP~$\s$, the class set is unique.
Moreover, we have $\ell(C_i) < \ell(C_{i+1})$ for all $i\ge 1$.
\end{corollary}

Since a client is added to a class the first time one of the facilities in her shopping range is added to a class, she is in the lowest class in her neighborhood.

\begin{corollary}[Clients in Lowest Class of Neighborhood]
\label{cor:client-min-class}
Each client $v$ is in class $C(v) = C_j$ with $j = \min\{ i \mid C_i \cap N_\s(v) \ne \varnothing \} $.
\end{corollary}

\citet{ijcai-21} give an algorithm to compute a MNS\footnote{We show in the \suppmaterial{}, why their algorithm is applicable despite the different definitions.} which we use repeatedly to compute the class set.
\begin{restatable}[Class Set Computation]{corollary}{classsetalg}
    \label{cor:class-set-alg}
    For a FPP~$\s$, the class set and the average loads may be computed efficiently.
\end{restatable}

We now introduce \emph{rounded} client profiles which we obtain by assigning the clients to facilities within each class $C_i$.
Note, that the definition only allows \emph{pure} client strategies.

\begin{definition}[Rounded Client Profile]
\label{def:rounded}
    For a FPP~$\s$, a \oneProfile{} $\sigma(\s)$ is called \emph{rounded} if for all $f \in F$ it holds that
    \[
        \ell_f(\s, \sigma) \in \left\{\left\lfloor \ell(C_i) \right\rfloor, \left\lceil \ell(C_i) \right\rceil \right\}
    \] and that each client is assigned to exactly one facility in her class.
    That is, for each $v \in V_i$ there is exactly one $g \in F_i$ with $\sigma(\s)_{v,g}=1$.
    We call a \fullProfile{} $\sigma$ \emph{rounded} if $\sigma(\s)$ is \emph{rounded} for each $\s \in S$.
\end{definition}

We observe that every rounded \oneProfile{} is stable.
\begin{lemma}
    A rounded \oneProfile{} is a \oneEq{}.
\end{lemma}
\begin{proof}
    Assume for the sake of contradiction that $\sigma(\s)$ is not a \oneEq{}.
    Hence, there exists some client $v$ that is assigned to facility $f$, and there is another facility $g \in N_\s(v)$ with $\ell_g(\s, \sigma) +1 < \ell_f(\s, \sigma)$.
    Let $C_i$ be the class of $v$ and $f$.
    Since $\ell_g(\s, \sigma) +1 < \ell_f(\s, \sigma)$, facility $g$ has to be in a class $C_j$ with $j < i$ and $\ell(C_j) < \ell(C_i)$.
    However, by definition of the class sets, $v$ should have been assigned to class $C_j$ and not $C_i$, which yields the contradiction.
\end{proof}

We now establish the existence of a rounding scheme that yields rounded client profiles, thereby proving their existence.
In particular, this implies that pure client equilibria exist.
\begin{theorem} [Existence of Rounded Profiles] \label{thm:rounded-eq-existence}
    For every unweighted instance of the atomic \flg{} and facility placement profile $\s$, there exists a rounded \oneProfile{} $\sigma(\s)$.
\end{theorem}
\begin{proof}
    We compute a rounded \oneProfile{} by appropriately assigning the clients to facilities for each class separately:
    
    For each class $C_i$, we define the directed, bipartite \emph{residual graph} $R_i$ based on a current (pure) \oneProfile{} $\sigma(\s)$.
    %(not necessarily feasible) \oneEq{}:
    Let $R_i=(V_i \cup F_i, E_i)$ with $(f,v)\in E_i$ if $\sigma(\s)_{v,f} = 0$ and $v \in A_\s(f)$, and the reverse edge in $(v,f)\in E_i$ if $\sigma(\s)_{v,f} = 1$.
    %The intuition of this residual graph is that it enables us to find augmenting paths to add another client to the \oneProfile{}.
    We start with all clients unassigned and repeat the following steps until all clients of $C_i$ are assigned:
    \begin{enumerate}
        \item Let $F^* \subseteq F_i$ be the facilities of $C_i$ with lowest current load.
        \item\label[step]{step-path} In the residual graph $R_i$, find a simple path from set $F^*$ to the set of unassigned clients in $V_i$, i.e., the result is a path between some facility of current minimal load and some unassigned client.
        \item\label[step]{step-augment} Use the path to augment the client assignment, i.e, if there is an edge $(f,v)$ in the path, then assign $v$ to $f$ and if there is an edge $(v,f)$, then unassign $v$ from $f$.
    \end{enumerate}
    \Cref{fig:rounded-existence} shows an intermediate state of an example run of this procedure.
    \Cref{step-augment} keeps the facility loads constant for all facilities except for the first one in the path.
    By the choice of $F^*$, we maintain the invariant that the maximum difference of facility loads is $1$.
    Hence in the end the loads must be $\lfloor \ell(C_i) \rfloor$ or $\lceil \ell(C_i) \rceil$ and the result is a rounded \oneEq{}.
    \begin{figure}
        \centering
        \begin{tikzpicture}
            % left
            \node[vert] (e1) [label=left:$e$] at (0*\trheight*0.95, 0.9*\trlength) {};
            \node[vert] (f1) [label=left:$f$] at (1*\trheight*0.95, 0.9*\trlength) {};
            \node[vert] (g1) [label=left:$g$] at (2*\trheight*0.95, 0.9*\trlength) {};
            \node[vert] (h1) [label=left:$h$] at (3*\trheight*0.95, 0.9*\trlength) {};
            \node[vert] (u1) at (0*\trheight*0.95, 0*\trlength) {};
            \node[vert] (u2) at (1*\trheight*0.95, 0*\trlength) {$u$};
            \node[vert] (u3) at (2*\trheight*0.95, 0*\trlength) {};
            \node[vert] (u4) at (3*\trheight*0.95, 0*\trlength) {$v$};
            
            \node[fac] at (e1) {};
            \node[fac] at (f1) {};
            \node[fac] at (g1) {};
            \node[fac] at (h1) {};
            
            % right
            \node[vert] (v1) at (7.5*\trheight * 0.5, 0*\trlength) {};
            \node[vert] (v2) at (8.5*\trheight * 0.5, 0*\trlength) {};
            \node[vert] (v3) at (9.5*\trheight * 0.5, 0*\trlength) {};
            \node[vert] (u)  at (10.5*\trheight * 0.5, 0*\trlength) {$u$};
            \node[vert] (v4) at (11.5*\trheight * 0.5, 0*\trlength) {};
            \node[vert] (v5) at (12.5*\trheight * 0.5, 0*\trlength) {};
            \node[vert] (v6) at (13.5*\trheight * 0.5, 0*\trlength) {};
            \node[vert] (v)  at (14.5*\trheight * 0.5, 0*\trlength) {$v$};
            
            \node[fac] (e) [label=left:$e$] at (4*\trheight, 0.9*\trlength) {};
            \node[fac] (f) [label=left:$f$] at (5*\trheight, 0.9*\trlength) {};
            \node[fac] (g) [label=left:$g$] at (6*\trheight, 0.9*\trlength) {};
            \node[fac] (h) [label=left:$h$] at (7*\trheight, 0.9*\trlength) {};
            
            % rectangles top
            \coordinate(f1a) at ($(e) - (0.25*\trheight + 0.25* \trheight - 0.1, 0.25*\trheight)$);
            \coordinate(f2a) at ($(e) - (0.25*\trheight + 0.25* \trheight + 0.0, 0.25*\trheight + 0.1)$);
            \coordinate(f3a) at ($(e) - (0.25*\trheight + 0.25* \trheight + 0.1, 0.25*\trheight + 0.2)$);
            \coordinate(f1b) at ($(f) + (0.25*\trheight + 0.25* \trheight, 0.6)$);
            \coordinate(f2b) at ($(g) + (0.25*\trheight + 0.25* \trheight, 0.7)$);
            \coordinate(f3b) at ($(h) + (0.25*\trheight + 0.25* \trheight + 0.1, 0.8)$);
            \draw[thick,dotted] (f1a) rectangle (f1b);
            \draw[thick,dotted] (f2a) rectangle (f2b);
            \draw[thick,dotted] (f3a) rectangle (f3b);
            \node[draw=none, anchor=north east] at (f1b) {$F^*$};
            \node[draw=none, anchor=north east] at (f2b) {$F^\text{all}$};
            \node[draw=none, anchor=north east] at (f3b) {$F_i$};

            % rectangles bottom
            \coordinate(v1a) at ($(u) - (0.28*\trheight, 0.8)$);
            \coordinate(v2a) at ($(e) - (0.25*\trheight + 0.25* \trheight + 0.1, 0.8 +\trlength)$);
            \coordinate(v1b) at ($(u) + (0.28*\trheight, 0.25*\trheight)$);
            \coordinate(v2b) at ($(h) + (0.25*\trheight + 0.25* \trheight + 0.1, 0.25*\trheight +0.1 - 0.9*\trlength)$);
            \draw[thick,dotted] (v1a) rectangle (v1b);
            \draw[thick,dotted] (v2a) rectangle (v2b);
            \node[draw=none, anchor=south west] at ($(v1a) - (0.13, 0)$) {$V^\text{all}$};
            \node[draw=none, anchor=south west] at (v2a) {$V_i$};

            % edges left
            \draw (u1) edge (e1);
            \draw (u2) edge (f1) edge[dashed] (g1);
            \draw (u3) edge (g1);
            \draw (u4) edge[dashed] (h1);
            
            % edges right
            \draw (e) edge[very thick,cred] (v1) edge (v2);
            \draw (f) edge (v3) edge[very thick,cred] (u);
            \draw (u) edge[very thick,cred, dashed] (g);
            \draw (g) edge[very thick,cred] (v4) edge (v5);
            \draw (h) edge (v6);
            \draw (v) edge[dashed] (h);
        \end{tikzpicture}
        \caption{Right: An intermediate state of the algorithm to compute a rounded \oneEq{} in \Cref{thm:rounded-eq-existence} for the instance on the left.
        The dashed edges denote current client assignments to facilities.
        The red paths mark two (out of many) augmenting paths for \Cref{step-path}.}
        \label{fig:rounded-existence}
    \end{figure}
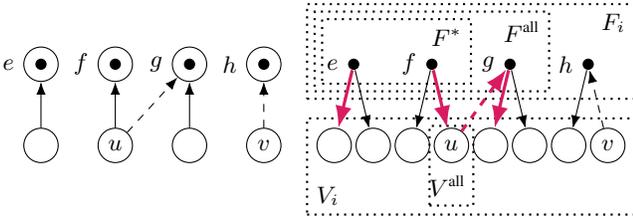

    It remains to show that the path in \Cref{step-path} always exists in each iteration.
    Assume towards contradiction that there is no such path, but an unassigned client.
    Let $F^\text{all}$ be the set of all facilities reachable from the set $F^*$ (with $F^* \subseteq F^\text{all}$) and let $V^\text{all}$ be the set of clients currently assigned to any facility in $F^\text{all}$.
    We have $V^\text{all} = A_\s(F^\text{all}) \cap V_i$ as there is no unassigned client reachable from $F^*$ and any assigned client in $A_\s(F^\text{all}) \cap V_i$ has an edge to her assigned facility, meaning that her facility is also in $F^\text{all}$.
    Now $F_i\setminus F^\text{all}\neq\varnothing$, because the unassigned client is in the attraction range of some facility in $F_i$.
    Since all facilities in $F_i\setminus F^\text{all}$ have a load higher than the average of $F^\text{all}$, we get
    \[
        \frac{w(A_\s(F^\text{all}) \cap V\setminus \bigcup_{j=1}^{i-1} V_j)}{|F^\text{all}|} < \frac{w(A_\s(F_i) \cap V\setminus \bigcup_{j=1}^{i-1} V_j)}{|F_i|}
    \]
    which contradicts $F_i$ being a minimum neighborhood set.
\end{proof}

We later use the lexicographic increasingly sorted vector of facility loads $\lsort$ as a potential function.
By construction of the class set and by definition of rounded client profiles, we immediately get the following.
\begin{observation}[Identical Load Vector for Rounded Profiles] \label{lem: sorted facility loads identical}
    For an instance of the unweighted atomic \flg{} and a facility placement profile $\s$, let $\sigma(\s)$ and $ \sigma'(\s)$ denote two rounded \oneProfile{}s for $\s$.
    Then, $\lsort(\s, \sigma) = \lsort(\s, \sigma')$.
\end{observation}

This enables switching between distinct rounded client equilibria without changing the value of the potential function $\lsort$.

\subsection{Favoring Client Equilibria}
We later seek to find a SPE by means of a sequence of improving facility moves.
However, it turns out, that it is not sufficient to do so with an arbitrary (rounded) \fullEq{}.
Due to differences in the rounding in $\sigma$ for different FPPs $\s$, there might exist an improvement sequence that cycles through a sequence of FPPs $\s_1,\s_2,\ldots,\s_m,\s_1$, while $\lsort(\s_i,\sigma)=\lsort(\s_j,\sigma)$ for all $i,j$.
See \Cref{fig:no-potential}.

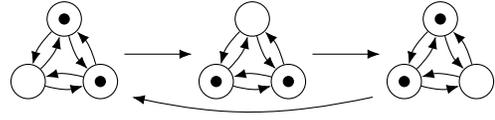
\begin{figure}
    \centering
    \begin{tikzpicture}
        \foreach \i / \name in {0/a,2.5/b,5/c} {
            \node[vert] (\name1) at (\i,0) {};
            \node[vert] (\name2) at (\i + 0.5*\trlength * 0.8, -\trheight * 0.8) {};
            \node[vert] (\name3) at (\i - 0.5*\trlength * 0.8, -\trheight * 0.8) {};
            \draw (\name1) to[bend right=15] (\name2);
            \draw (\name1) to[bend right=15] (\name3);
            \draw (\name2) to[bend right=15] (\name1);
            \draw (\name2) to[bend right=15] (\name3);
            \draw (\name3) to[bend right=15] (\name1);
            \draw (\name3) to[bend right=15] (\name2);
        }
        \node[fac] at (a1) {};
        \node[fac] at (a2) {};
        \node[fac] at (b3) {};
        \node[fac] at (b2) {};
        \node[fac] at (c3) {};
        \node[fac] at (c1) {};
        
        \draw (0.8, -0.45*\trheight) to (1.7, -0.45*\trheight);
        \draw (3.3, -0.45*\trheight) to (4.2, -0.45*\trheight);
        \draw (4.1, -1*\trheight) to[bend left=12] (0.9, -1*\trheight);
    \end{tikzpicture}
    \caption{The following rounded \fullEq{} produces the given best response cycle for the facilities in this instance: Clients prefer locations in clockwise order so that the facility in the preferred location receives load $2$ and the other facility load $1$.} \label{fig:no-potential}
\end{figure}
To circumvent cycling, we need to ensure that rounding is consistent between FPPs.
We solve that by making sure that the clients favor facilities based on a specifically chosen ordering of the facilities.
Therefore, for a given permutation $\pi$ of the set of facilities, let us first define a $\pi$-\emph{favoring} rounded equilibrium as follows.
\begin{definition} [$\pi$-\emph{favoring}] \label{def:favoring}
    For a fixed permutation $\pi$ of the set of facilities, let $\lpi:=\left(\ell_{\pi(1)}(\s, \sigma),\ldots,\ell_{\pi(k)}(\s, \sigma)\right)$.
    A rounded \oneEq{} $\sigma(\s)$ is $\pi$-\emph{favoring} if it maximizes the vector $\lpi$ lexicographically among all rounded \oneProfile{}s for $\s$.
    Similarly, a \fullEq{} $\sigma$ is $\pi$-favoring if $\sigma(\s')$ is a $\pi$-favoring \oneEq{} for every $\s' \in S$.
\end{definition}

As our key technical ingredient, we next show that any improving facility move from $(\s,\sigma$) increases the vector $\lsort$ lexicographically if $\sigma$ is $\pi$-favoring for a specific permutation $\pi$.
In particular, for this $\pi$ the entries of $\lpi$ are non-increasing.

\begin{lemma} [Lexicographical Potential Increase]\label{lem:lex-increase}
Let $\s$ be a facility placement profile.
Every improving facility move away from $(\s, \sigma)$ strictly increases $\lsort$ lexicographically if $\sigma$ is $\pi$-favoring for a permutation $\pi$ with the following property:
For all $f, g \in F$, if $\ell_f(\s, \sigma) > \ell_g(\s, \sigma)$, then $\pi^{-1}(f)<\pi^{-1}(g)$.
\end{lemma}

\begin{proof}
    Assume for the sake of contradiction that some facility $f$ can improve, changing the state from $\s$ to $\s'$ but without increasing $\lsort$ lexicographically.
    Let $\C=\{C_1,C_2,\ldots\}$ and $\C'=\{C'_1,C'_2,\ldots\}$ be the class sets for $\s$ and $\s'$, respectively.
    
    Since the load of $f$ increased but $\lsort$ did not increase lexicographically, there exists a facility $d \in F$ with
    \begin{equation}
        \ell_d(\s', \sigma) < \ell_d(\s, \sigma) \text{ and } \label{E1}
        \ell_d(\s', \sigma) \leq \ell_f(\s, \sigma)\text.
    \end{equation}
    As $\sigma$ is a rounded \fullProfile{}, by \Cref{def:rounded}, we have
    % \begin{equation}
    %     \ell_d(\s,\sigma)\le \lceil \ell(C(d)) \rceil\text{ and }
    %     \lfloor \ell(C'(d)) \rfloor \le \ell_d(\s',\sigma)\text.\label{E2}
    % \end{equation}
    % Combining (\ref{E1}) and (\ref{E2}) yields
    \begin{equation}
        \label{E3}
        \lfloor \ell(C'(d)) \rfloor \le \ell_d(\s',\sigma) \le \ell_d(\s,\sigma) - 1 \le \lfloor \ell(C(d) \rfloor \text.
    \end{equation}
    For our analysis, we construct the following \emph{reassignment graph} $G$ that captures the difference in client assignment between $\sigma(\s)$ and $\sigma(\s')$.
    This is a directed, multigraph $G=(F,E)$ where we add an edge $(f_1,f_2)$ for every $v$ with $\sigma(\s)_{v,f_1} =1$ and $\sigma(\s')_{v,f_2} =1$.
    As the load of $d$ has decreased, the out-degree of $d$ is strictly larger than her in-degree in $G$.
    Therefore, there exists a directed path $p$ from $d$ to a facility $g$ which has an in-degree that is strictly larger than her out-degree.
    Comparing the assignments of the clients that determine the edges of $p$ between $\sigma(\s)$ and $\sigma(\s')$, we observe that the difference only affects the loads of $d$ and $g$.
    Hence for this facility $g$, we have
    \begin{equation} \label{E4}
        \ell_g(\s,\sigma) < \ell_g(\s',\sigma)\text.
    \end{equation}
    Again, by \Cref{def:rounded}, we have
    % \begin{equation} \label{E5}
    %     \lfloor \ell(C(g)) \rfloor \le \ell_g(\s,\sigma)\text{ and }
    %     \ell_g(\s',\sigma)\le \lceil \ell(C'(g)) \rceil \text.
    % \end{equation}
    % Combining (\ref{E4}) and (\ref{E5}) gives
    \begin{equation} \label{E6}
        \lfloor \ell(C(g)) \rfloor \le \ell_g(\s,\sigma) \le \ell_g(\s',\sigma) -1 \le \lfloor \ell(C'(g)) \rfloor \text.
    \end{equation}
    For any two consecutive facilities $a,b$ on the path $p$, there is a client $v$ that is assigned to $a$ in $\s$ and $b$ in $\s'$.
    By \Cref{cor:client-min-class} we conclude that $\ell_a(\s, \sigma)\le \ell_b(\s, \sigma)$ and $\ell_a(\s', \sigma)\ge \ell_b(\s', \sigma)$.
    This implies for the average load of the respective classes $\ell(C(a)) \le \ell(C(b))$ and $\ell(C'(a)) \ge \ell(C'(b))$ which yields
    \begin{equation}
        \ell(C(d)) \le \ell(C(g)) \text{ and } \label{E7}
        \ell(C'(d)) \ge \ell(C'(g)) \text.
    \end{equation}
    We now distinguish two cases.
\begin{description}[leftmargin=0.4cm]
\item[Case 1:] $g \ne f$.
    Combining (\ref{E6}) with (\ref{E3}) and (\ref{E7}) results in
    \begin{equation} \label{E8}
    \begin{split}
        \lfloor \ell(C'(d)) \rfloor \le \lfloor \ell(&C(d)) \rfloor \le \lfloor \ell(C(g)) \rfloor \\&\le \lfloor \ell(C'(g)) \rfloor
        \le \lfloor \ell(C'(d)) \rfloor \text.
    \end{split}
    \end{equation}
    And, hence, the above holds with equality.
    By (\ref{E1}), (\ref{E4}), and \Cref{def:rounded}, we get that
    \begin{equation} \label{E9}
        \ell_g(\s,\sigma) +1 = \ell_d(\s',\sigma) +1 = \ell_g(\s',\sigma) = \ell_d(\s,\sigma).
    \end{equation}

    By (\ref{E8}) and (\ref{E9}), we have that both assignments of the clients that determine the path $p$ yield a feasible rounded \oneEq{} for both $\s$ and $\s'$.
    As $d$ is preferred in $\sigma(\s)$, we must have $\pi^{-1}(d) > \pi^{-1}(g)$, however, from $\sigma(\s')$, we get that $\pi^{-1}(g) > \pi^{-1}(d)$, which yields the contradiction.
    
\item[Case 2:] $g=f$.
    Using (\ref{E7}), \Cref{def:rounded}, and (\ref{E1}) we have
    \begin{equation} \label{E10}
        \lfloor \ell(C'(f)) \rfloor \le \lfloor \ell(C'(d)) \rfloor \le \ell_d(\s',\sigma)\le \ell_f(\s,\sigma)\text.%\le \lceil \ell(C(f)) \rceil
    \end{equation}
    By $f$ improving from $\s$ to $\s'$ and \Cref{def:rounded}:
    \begin{equation} \label{E11}
        %  \lfloor \ell(C(f)) \rfloor \le
        \ell_f(\s,\sigma) < \ell_f(\s',\sigma) \le \lfloor \ell(C'(f)) \rfloor +1\text,
    \end{equation}
    which implies equality for
    \begin{equation} \label{E12}
    \begin{split}
        \lfloor \ell(C'(f)) \rfloor
        &= \lfloor \ell(C'(d)) \rfloor = \ell_d(\s',\sigma) \\&= \ell_f(\s,\sigma) =\ell_f(\s',\sigma) -1.
    \end{split}
    \end{equation}
    Together with (\ref{E1}) we get
    $\ell_d(\s, \sigma) > \ell_f(\s, \sigma)$ which implies $\pi^{-1}(d) < \pi^{-1}(f)$ by the additional property of the lemma.
    This, however, contradicts $\sigma(\s')$ being $\pi$-favoring.
    In $\sigma(\s')$, assigning the clients that determine path $p$ as in $\sigma(\s)$ lexicographically increases $\lpi$ since the load of $d$ increases by $1$ while the load of $f$ decreases by $1$.
    
    Observe that this still yields a feasible rounded \oneProfile{} as the loads of all facilities except $d$ and $f$ did not change and we have that $\ell_f(\s',\sigma) = \lfloor \ell(C'(f)) \rfloor +1$ and since $\ell(C'(d)) \ge \ell(C'(f))$ and $\lfloor \ell(C'(d)) \rfloor =\ell_d(\s',\sigma)$ also $\ell_d(\s',\sigma) +1=\lceil \ell(C'(d)) \rceil$.
\end{description}
    In both cases we reached a contradiction, hence such a facility $d$ cannot exist.
    The lemma follows.
\end{proof}

The existence of $\pi$-favoring \oneEqa{} and $\pi$-favoring \fullEqa{} follows directly from \Cref{thm:rounded-eq-existence} and \Cref{def:favoring}.
Moreover, we show that a $\pi$-favoring \fullEq{} can be computed in polynomial time.
\Cref{alg:favoring} computes $\pi$-favoring \oneEqa{} by employing a reduction to integral \textsc{MaxCostFlow}, where for each client a demand of $1$ flows to the facilities of her class in her shopping range and then to a sink node $t$.

\begin{algorithm}[t]
    \caption{findFavoringEquilibrium($H$, $\s$, $\pi$)}
    \label{alg:favoring}
    $G \gets (F \cup V \cup \{s, t\}, E_N \cup E_s \cup E_t \cup E_\text{add})$\;
    $E_s \gets \{(s, v, 1, 0) \mid v \in V\}$\;
    $E_N \gets \{(v, f, 1, 0) \mid C(v) = C(f) \wedge f \in N_\s(v)\}$\;
    $E_t \gets \{(f, t, \lfloor \ell(C(f)) \rfloor, 2^k) \mid f \in F\}$\;
    $E_\text{add} \gets \{(f, t, 1, 2^{k - \pi^{-1}(f)}) \mid f \in F\}$\;
    compute integral $s$-$t$-\textsc{MaxCostFlow} on $G$\;
    \For{$v \in V$}{
        assign $v$ to facility $f$ with nonzero flow on $(v,f)$\;
    }
\end{algorithm}

\Cref{fig:example-favoring-alg} shows an example of the construction.
Each facility $f$ has one edge with capacity $\lfloor \ell(C(f)) \rfloor$ and cost $2^k$ to $t$ and a second edge to $t$ with capacity $1$ and cost $2^{k - \pi^{-1}(f)}$.
This ensures that each facility receives an allowed value of load for a rounded \oneEq{}, while the cost of the edges with capacity $1$ prefers facilities appearing earlier in $\pi$.

\begin{figure}
    \centering
    \begin{tikzpicture}
        \node[vert] (v1a)[label=left:{$v_1$}][label=right:{$f_1$}] at (-1.9*\trheight,0.8*\trheight) {};
        \node[vert] (v2a)[label=left:{$v_2$}][label=right:{$f_2$}] at (-1.9*\trheight,0) {};
        \node[vert] (v3a)[label=left:{$v_3$}] at (-1.9*\trheight,-0.8*\trheight) {};
        \node[fac] at (v1a) {};
        \node[fac] at (v2a) {};

        \draw (-1.05*\trheight,0) to (-0.55*\trheight,0);
        
        \node[vert, rectangle, draw=none] (s1) at (0, 0.8*\trheight) {};
        \node[vert, rectangle, draw=none] (s3) at (0, -0.8*\trheight) {};
        \node[vert, rectangle, draw=none] (t1) at (5.2*\trheight, 0.55*\trheight) {};
        \node[vert, rectangle, draw=none] (t3) at (5.2*\trheight,-0.55*\trheight) {};
        \node[vert, rectangle, rounded corners, minimum height = 13pt + 1.6*\trheight cm] (s2) at (0,0) {$s$};
        \node[vert, rectangle, rounded corners, minimum height = 13pt + 1.6*\trheight cm] (t2) at (5.2*\trheight,0) {$t$};
        \node[vert] (v1) at (1.3*\trheight,0.8*\trheight) {$v_1$};
        \node[vert] (v2) at (1.3*\trheight,0) {$v_2$};
        \node[vert] (v3) at (1.3*\trheight,-0.8*\trheight) {$v_3$};
        \node[vert] (f1) at (2.6*\trheight,0.55*\trheight) {$f_1$};
        \node[vert] (f2) at (2.6*\trheight,-0.55*\trheight) {$f_2$};
        \draw (v2a) edge(v1a);
        \draw (v3a) edge(v2a);

    \begin{scope}
        \draw (s1) edge["{$1,0$}"] (v1);
        \draw (s2) edge["{$1,0$}"] (v2);
        \draw (s3) edge["{$1,0$}"] (v3);
        \draw (f1) to[bend left=7, "{$\lfloor \ell(C(f_1)) \rfloor, 2^k$}"] (t1);
        \draw (f1) to[bend right=7, "{$1, 2^{k - \pi^{-1}(f_1)}$}" below] (t1);
        \draw (f2) to[bend left=7, "{$\lfloor \ell(C(f_2)) \rfloor, 2^k$}"] (t3);
        \draw (f2) to[bend right=7, "{$1, 2^{k - \pi^{-1}(f_2)}$}" below] (t3);
        \draw (v1) edge["{$1,0$}"] (f1);
        \draw (v2) edge["{$1,0$}"] (f1);
        \draw (v2) edge["{$1,0$}"] (f2);
        \draw (v3) edge["{$1,0$}"] (f2);
    \end{scope}
    \end{tikzpicture}
    \caption{\Cref{alg:favoring} creates the graph $G$ on the right from the host graph and facility placement profile on the left.
    The edges are labeled with capacity, cost and $k$ is the number of facility agents.}
    \label{fig:example-favoring-alg}
\end{figure}
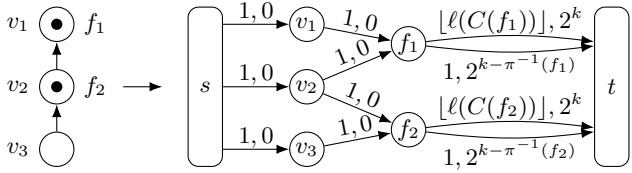

\begin{theorem}[$\pi$-favoring Equilibria Complexity]
    For any \textup{FPP} $\s$ and any $\pi$, a $\pi$-favoring \oneEq{} $\sigma$ can be computed in polynomial time.
\end{theorem}
\ifdefined\arxiv
\begin{proof}
    First, we show that the result $\sigma(\s)$ of the algorithm is a rounded \oneProfile{}.
    
    Note that since we compute an integral flow, each client is assigned to only one facility.
    By the setup of $E_s$, a client only puts weight on facilities that are both in her class and in her shopping range.
    Observe that the maximum cost contributed by edges in $E_\text{add}$ is $2^k-1$ smaller than sending a single unit of flow through an edge in $E_t$.
    There cannot be a facility that receives a load of less than $\lfloor \ell(C(f)) \rfloor$, because then the total flow cost is lower than $\sum_{f \in F}{ \lfloor \ell(C(f)) \rfloor}2^k$.
    Converting an arbitrary rounded \oneProfile{} (which is guaranteed to exist by \Cref{thm:rounded-eq-existence}) to a flow achieves a cost of at least $\sum_{f \in F}{ \lfloor \ell(C(f)) \rfloor}2^k$.
    If $\ell(C(f))$ is an integer, then no clients are leftover to provide flow through edges in $E_\text{add}$ and otherwise, these edges limit the load of a facility to $\lceil \ell(C(f)) \rceil$.

    Assume towards contradiction that there is a rounded \oneProfile{} $\sigma'(\s)$ for which $\lpi(\sigma) >_\text{lex} \lpi(\sigma)$.
    Let $i$ be the first spot in which these vectors differ with the facility $f=\pi(i)$ receiving load $\lfloor \ell(C(f)) \rfloor$ in $\sigma(\s)$ and $\lceil \ell(C(f)) \rceil$ in $\sigma'(\s)$.
    The corresponding flows have the same cost for edges in $E_t$ and edges $(g,t)$ in $E_\text{add}$ with $\pi^{-1}(g) < i$.
    For $\sigma'(\s)$ the edge $(f,t)$ provides a cost of $2^{k-i}$ higher than the maximum possible total cost $2^{k-i}-1$ of all remaining edges $E_\text{add}$ with $\pi^{-1}(g) > i$, so the flow corresponding to $\sigma(\s)$ was not a maximum cost flow.

    The algorithm needs the class set for the construction of the graph $G$, which can be computed in polynomial time by \Cref{cor:class-set-alg}.
    A \textsc{MaxCostFlow} can be computed using the algorithm by \citet{mincostflow}.
\end{proof}
\else
\noindent A detailed proof is in the \suppmaterial{}.
\fi

\subsection{Constructing a Subgame Perfect Equilibrium}
We now prove the existence of SPEs by providing an algorithm that computes one.
We start with an arbitrary FPP~$\s$ and \fullProfile{} $\sigma$.
If $(\s,\sigma)$ is not a SPE, i.e., there is a facility that can improve resulting in a profile $(\s',\sigma)$, we determine a possibly different \fullProfile{} $\sigma'$ that satisfies the conditions of \Cref{lem:lex-increase} for a possibly different $\pi'$.
We repeat this process until we find a SPE.
This process terminates as by \Cref{lem:lex-increase} in each step the value of $\lsort$ increases.

\begin{theorem}[Existence of SPE] \label{thm:unweighted-spe}
    Every instance of the unweighted atomic \flg{} admits a subgame perfect equilibrium with a pure \fullEq{}.
\end{theorem}

\begin{algorithm}[t]
    \caption{findSPE(H, U, k)}
    \label{alg:spe}
    $\s \gets$ arbitrary facility placement profile\;
    $\sigma \gets$ arbitrary rounded \fullProfile{}\;
    \While{$(\s, \sigma)$ is not a SPE}{
        $\s \gets$ execute arbitrary improving facility move\;
        \label{line:fac-move}
        $\pi \gets$ facilities sorted by decreasing loads in $\sigma(\s)$\;
         \label{line:pi}
        $\sigma \gets$ a $\pi$-favoring \fullEq{}\;
         \label{line:sigma}
    }
\end{algorithm}
\begin{proof}
    We show that \Cref{alg:spe} computes a SPE.
    By \Cref{thm:rounded-eq-existence}, an initial rounded profile $\sigma$ exists.
    A $\pi$-favoring \fullEq{} exists by \Cref{thm:rounded-eq-existence} and \Cref{def:favoring}.
    It remains to show that this algorithm terminates.
    To that end, we show that in each iteration except the first, the vector of sorted facility loads increases lexicographically.
    
    The conditions of \Cref{lem:lex-increase} are satisfied by the permutation~$\pi$ in \Cref{line:pi}, so for every improving facility move in \Cref{line:fac-move} with the facility placement profiles $\s$ and $\s'$ before and after the move, the value of $\lsort$ lexicographically increases:
    \[
        \lsort(\s', \sigma) >_\text{lex} \lsort(\s, \sigma).
    \]
    By \Cref{lem: sorted facility loads identical}, $\lsort$ does not change when switching to a different rounded \fullEq{} in \Cref{line:sigma} while keeping the FFP $\s$ fixed.
    Since the number of possible facility placement profiles is finite, the algorithm terminates.
\end{proof}

Regarding the complexity of computing a SPE, we note that a \fullEq{} requires exponential space because there are exponentially many FPPs for which clients need to specify a client profile.
However, for certifying the existence of a SPE and for \Cref{alg:spe}, it is sufficient to specify the \oneEqa{} for the chosen state $\s$ and all states $\s'$ with only a single facility deviation.
With that caveat, each iteration of \Cref{alg:spe} can be performed in polynomial time.
However, the number of iterations needed is an open problem.

\section{SPE in the Weighted Game}
\label{sec:weighted}

We now turn to instances with weighted clients.
Here, SPE do not exist in all instances and it is even NP-complete to decide whether an approximate SPE exists that has a ratio smaller than the golden ratio $\phi$.
In \Cref{fig:no-spe}, we show a simple instance with no SPE and another one that does not admit a $(\phi-\epsilon)$-approximate SPE.
In the \suppmaterial{}, we give detailed proof of the non-existence.

\begin{restatable}[Nonexistence of Weighted SPE]{theorem}{nospe}
Weighted instances of the atomic \flg{} might not admit SPE.
\end{restatable}

\begin{restatable}[Approximate Weighted SPE Lower Bound]{theorem}{noapproxspe}
\label{thm:no-approx-spe}
Weighted instances of the atomic \flg{} might not admit $(\phi-\epsilon)$-approximate SPE for any $\epsilon>0$.
\end{restatable}

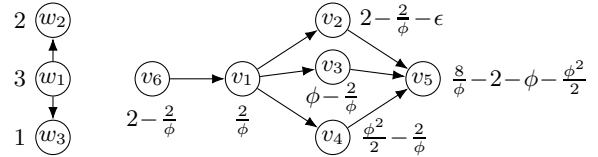
\begin{figure}
    \centering
    \begin{tikzpicture}
        \medmuskip=1mu
        \node (1) [vert,label=below:{$\frac{2}{\phi}$}]{$v_1$};
        \node (3) [vert,label={[label distance=-0.15cm]below:{$\phi-\frac{2}{\phi}$}}] at (1*\trlength, 0.15*\trheight) {$v_3$};
        \node (2) [vert,label={[inner ysep=0pt]right:{$2-\frac{2}{\phi}-\epsilon$}}] at (1*\trlength, 0.75*\trheight) {$v_2$};
        \node (4) [vert,label={[inner ysep=0pt]2:{$\frac{\phi^2}{2} - \frac{2}{\phi}$}}] at (1*\trlength, -0.75*\trheight) {$v_4$};
        \node (5) [vert,label=right:{$\frac{8}{\phi} - 2 - \phi -\frac{\phi^2}{2}$}] at (2*\trlength, 0) {$v_5$};
        \node (6) [vert,label=below:{$2 - \frac{2}{\phi}$}, left of=1]{$v_6$};
        \draw (1) -- (2);
        \draw (1) -- (3);
        \draw (1) -- (4);
        \draw (2) -- (5);
        \draw (3) -- (5);
        \draw (4) -- (5);
        \draw (6) -- (1);
        
        \node (1a) [vert,label=left:$3$] at (-2.1*\trlength, 0) {$w_1$};
        \node (2a) [vert,label=left:$2$] at (-2.1*\trlength, 0.75*\trheight) {$w_2$};
        \node (3a) [vert,label=left:$1$] at (-2.1*\trlength, -0.75*\trheight) {$w_3$};
        \draw (1a) -- (2a);
        \draw (1a) -- (3a);
    \end{tikzpicture}
    \caption{For two facility agents, an instance without SPE (left) and an instance without $(\phi-\epsilon)$-approximate SPE (right).}
    \label{fig:no-spe}
\end{figure}

It is open if this bound is tight.
Note that a $k$-approximate SPE always exists:
Place each facility on the single vertex with the highest total client weight in her attraction range, and let the clients use equal probabilities for each facility.

\begin{observation}[Approximate SPE Upper Bound]
\label{obs:approx-spe-upper}
Each instance of the \flg{} admits a $k$-approximate SPE.
\end{observation}

We can even show that deciding whether a $(\phi-\epsilon)$-approximate SPE exists is NP-complete via a reduction from SAT.
Membership in NP is given by the fact that a FPP and \oneEqa{} for only the polynomially many neighboring FPPs are sufficient to verify the existence of a SPE.
\ifdefined\arxiv\else
We give detailed proof in the \suppmaterial{}.
\fi

\begin{theorem}[Complexity of Approximate Weighted SPE]
\label{thm:npc}
    Let $\alpha \in [1, \phi)$ denote some multiplicative approximation ratio.
    It is NP-complete to decide whether a weighted instance of the atomic \flg{} admits an $\alpha$-approximate SPE.
\end{theorem}
\ifdefined\arxiv
\begin{proof}
    Let $\alpha \in (1, \phi)$ denote some approximation ratio.
    We reduce from SAT.
    Let an instance of SAT consisting of a set of $t$ clauses $C$ and a set of $m$ binary variables $\mathbf{x} = \{x_1, \ldots, x_m\}$, where we assume $t \geq 4$ for simplicity.
    We construct an instance $(G_1 \cup G_2 \cup G_3, U, k)$ of \flg{} which depends on the instance of SAT.
    We create $G_1 = (V_1, E_1)$ with $V_1 = V^\text{yes} \cup V^\text{no} \cup C \cup B$ with $|B| = (m-1)t$, $V^\text{yes} = \{y_1, \ldots, y_m\}$ and $V^\text{no} = \{n_1, \ldots, n_m\}$.
    We set all weights for $v \in V_1$ to $w(v) = \frac{m}{m(t+2) - 1}$.
    $E_1$ contains the following edges:
    \begin{align*}
        (n_i, y_i), (y_i, n_i) & \text{ for } i = 1, \ldots, m\text,\\
        (b, y) & \text{ for each } b, y \in B \times V^\text{yes,}\\
        (b, n) & \text{ for each } b, n \in B \times V^\text{no,}\\
        (c, y_i) & \text{ for each literal } x_i \in c, c \in \mathbf{x}\text{ and}\\
        (c, n_i) & \text{ for each literal } \neg x_i \in c, c \in \mathbf{x}\text.
    \end{align*}
    For $G_2$ we use the graph on the right in \Cref{fig:no-spe}, which does not admit any $(\phi-\epsilon)$-approximate SPE for two facility agents (\Cref{thm:no-approx-spe}) and choose a sufficiently small $\epsilon$.
    For $G_3$ we use the graph in \Cref{fig:component3}.
    
    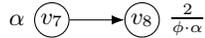
\begin{figure}[h]
        \centering
        \begin{tikzpicture}
            \node (7) [vert,label=left:{$\alpha$}]{$v_7$};
            \node (8) [vert,label=right:{$\frac{2}{\phi \cdot \alpha}$}, right of=7]{$v_8$};
            \draw (7) -- (8);
        \end{tikzpicture}
        \caption{Component $G_3$ of the weighted instance reduced from SAT in \Cref{thm:npc}.}
        \label{fig:component3}
    \end{figure}

We use the set of facilities $F = Q \cup \{g, h\}$ with the set $Q$ consisting of $m$ facilities and with the following restrictions:
\begin{align*}
    \forall q \in Q : U(q) &= V_1 \cup \{v_7\}\text,\\
    U(g) &= \{v_1, v_2, v_3, v_4, v_5, v_6\}\text{ and}\\
    U(h) &= \{v_1, v_2, v_3, v_4, v_5, v_6, v_8\}\text.
\end{align*}
We now show that this instance admits an $\alpha$-approximate SPE if and only if there is a satisfying assignment to the variables of $\mathbf{x}$.
Assume that there exists a satisfying assignment to the variables of $\mathbf{x}$ and let $\mathbf{z}$ denote one such assignment.
Then we set the facility placement profile $\s$ such that $s_g = v_1$, $s_h=v_8$ and for each $z_i \in \mathbf{z}$ we assign a facility $q \in Q$ to
\[
    s_q = \begin{cases}
        y_i & \text{ if } z_i \text{ is true,} \\
        n_i & \text{ if } z_i \text{ is false}.
    \end{cases}
\]
We construct a pure \oneProfile{} $\sigma(\s)$ such that all facilities in $Q$ receive equal load.
This is possible because the number of clients in $V_1$ is completely covered by facilities in $Q$, is divisible by $|Q| = m$ and we chose the size of $B$ which is accessible for all facilities placed in $V_1$ sufficiently large to equalize skewed distributions of the clients in $C$.
We also set $\sigma$ such that a facility $q \in Q$ which individually deviates to a vertex in $V_1$ receives the same utility as in $(\s, \sigma)$, which we can control through the assignment of the clients in $B$ again.

To prove that $(\s, \sigma)$ is an $\alpha$-approximate SPE, it suffices to show that none of the facility agents can improve by a factor $\alpha$ by moving away from $\s$ in the facility game induced by $\sigma$.
First note that $g$ and $h$ cannot improve, since the maximal available utility of every vertex they can move to is smaller than their current load.
Next, note that the sum of the loads of facilities the facilities in $Q$ is $m(t+2) \cdot \frac{m}{m(t+2)-1} > m$ and, thus, the loads on facilities in $Q$ are strictly larger than $1$.
This means that a facility agent $q \in Q$ moving to $v_7$ improves by a factor smaller than $\alpha$.
Thus, no facility can improve by a factor of at least $\alpha$, so we conclude that $(\s, \sigma)$ is an $\alpha$-approximate SPE.

Now assume that there exists no satisfying assignment for the instance of SAT, and assume by contradiction that an $\alpha$-approximate SPE $(\s, \sigma)$ exits.
\Cref{thm:no-approx-spe} shows that no $\alpha$-approximate facility equilibrium exists with exactly two facilities located on the vertices of $G_2$.
Since $g$ and $h$ are the only facilities able to pick any vertex on $G_2$ as their location, and $g$ can \emph{only} pick these vertices, $h$ must be on $s_h = v_8$.

Next, note that $s_g \neq v_5$, as this would allow $g$ to improve by a factor $\phi > \alpha$ by moving to $v_1$.
It follows that $h$ can attain a load of $\frac{2}{\phi}$ by moving to $v_{5}$.
By assumption, this move does not increase the load on $h$ by a factor $\alpha$, i.e., $\ell_h(\s, \sigma) > \frac{2}{\phi \cdot \alpha}$.
Thus, client $v_7$ puts weight on $f_{h}$, which implies that no facility is located on $v_7$.
We conclude that facilities in $Q$ are all located on $V_1$.

Since there is no satisfying assignment for the instance of SAT, a set of $m$ facilities cannot cover all clients in $V_1$ for any possible \oneProfile{}.
The sum of the loads on facilities in $Q$ is therefore at most $(m(t+2) - 1) \cdot \frac{m}{m(t+2) -1} = m$.
Thus, there is some facility agent $q \in Q$ with load at most $1$, who can improve her load by a factor $\alpha$ by moving to $v_7$.
Hence, $\s$ is not an $\alpha$-approximate facility equilibrium; a contradiction.
We conclude that the problem of deciding whether an instance of \flg{} admits an $\alpha$-approximate SPE is NP-hard.

Even though the size of a SPE is exponential, a partial approximate SPE $(\s, \sigma^\text{partial})$ of polynomial size is sufficient for verification, where $\sigma^\text{partial}$ contains only the client equilibria for $\s$ and the facility placement profiles with exactly one facility deviation from $\s$.
As established in \Cref{sec:client-eq}, all other facility placement profiles also have client equilibria.
Verifying that $(\s, \sigma^\text{partial})$ is in fact a partial approximate SPE may be done in polynomial time by testing all possible facility deviations and checking if \Cref{obs:excluded-load} holds.
\end{proof}
\fi

\section{Equilibrium Efficiency}
\label{sec:poa}
Note that the weighted participation rate captures both the social welfare of the clients and the facilities.
Note that for the clients the waiting time is not taken into account, only whether they are served at all.
% Expand in journal version
We prove that the PoA is exactly $2$.
This improves the upper bound of \citet{ijcai-21} that only holds for instances where each facility may choose any location.
\ifdefined\arxiv
In comparison to their Theorem 6, our lower bound is tight, which relies on the specific behavior of our atomic clients, however.
\else
Also, as another contrast, in the \suppmaterial{} we show a matching lower bound.
There, we also show how to transfer their lower bound of $2-\frac1k$ on the price of stability to our model.
\fi

\begin{theorem}[Price of Anarchy]
\label{thm:poa}
For the weighted atomic \flg{}, the price of anarchy regarding the utilitarian social welfare of the facility agents is $2$ for $k \geq 2$ facility agents.
\end{theorem}
\ifdefined\arxiv\begin{proof}\else\begin{proof}[Proofsketch]\fi
    We consider an instance $(H, U, k)$ of \flg{} with $k\geq 2$, and let $(\s, \sigma)$ denote some SPE.
    Let $V^\text{un}(\s)$ denote the set of uncovered clients with no facility in their shopping range for $\s$.
    We say a subset $V' \subseteq V^\text{un}(\s)$ is an \emph{$f$-cluster} if they can all be covered by facility $f$ at the same time, i.e., if there exists a vertex $v \in U(f)$ such that $V' \subseteq N(v)$.
    Since $\s$ is a facility equilibrium, the maximum weight of an $f$-cluster is $\ell_f(\s, \sigma)$, as otherwise, $f$ could improve by at least the complete weight of the $f$-cluster by moving to $v$.
    
    For each $f \in F$, let $V^f$ denote an $f$-cluster of maximal weight.
    Then the total weight of clients in $V^\text{un}(\s)$ covered by an alternative FPP is at most $\sum_{f \in F} w(V^f)$.
    Any such alternative FPP may still cover all previously covered clients, so we bound $W(\text{OPT})$ as follows:
    $W(\text{OPT})\leq$
    
    \noindent\resizebox{\linewidth}{!}{\( \displaystyle
          \sum_{f \in F} \left(\ell_f(\s, \sigma) + w(V^f)\right) \leq \sum_{f \in F} \left(\ell_f(\s, \sigma) + \ell_f(\s, \sigma)\right) = 2 W(\s)\text.
    \)}
    Therefore, the price of anarchy is at most $2$.
    \ifdefined\arxiv
    
    Now, we show that this bound is tight even for unweighted and unrestricted instances.
    To this end, let $(H, k)$ be an unrestricted instance of unweighted atomic \flg{} with $k \geq 2$
    We set the host graph to be $H = (V^\text{core} \cup V^\text{out}, A)$ with $V^\text{core} = \{c_1, \ldots, c_k\}$, $V^\text{out} = \{o_1, \ldots, o_k\}$ and
    \[
        A = \{(c_i, c_j)\mid i,j \leq k, i \neq j\} \cup \{c_i, o_i\mid i \leq k\}\text.
    \]
    The host graph for $k=3$ is shown in \Cref{fig:example-efficiency}.
    
    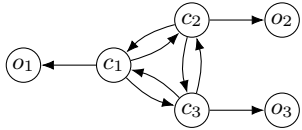
\begin{figure}[h]
        \centering
        \begin{tikzpicture}
            \node[vert] (c1) at (0,0) {$c_1$};
            \node[vert] (c2) at (\trheight, 0.5*\trlength) {$c_2$};
            \node[vert] (c3) at (\trheight, -0.5*\trlength) {$c_3$};
            \node[vert] (o1) [left of=c1]{$o_1$};
            \node[vert] (o2) [right of=c2]{$o_2$};
            \node[vert] (o3) [right of=c3]{$o_3$};
    
            \draw (c1) to[bend right=15] (c2);
            \draw (c1) to[bend right=15] (c3) ;
            \draw (c2) to[bend right=15] (c1);
            \draw (c2) to[bend right=15] (c3);
            \draw (c3) to[bend right=15] (c1);
            \draw (c3) to[bend right=15] (c2);
            \draw (c1) to (o1);
            \draw (c2) to (o2) ;
            \draw (c3) to (o3);
        \end{tikzpicture}
        \caption{The host graph for the lower bound in \Cref{thm:poa} for $k=3$ facility agents.} \label{fig:example-efficiency}
    \end{figure}
    
    Consider facility placement profiles $\s_\text{opt} := (o_1, o_2, \ldots, o_k)$ and $\s = (v_1, v_2, \ldots, v_k)$.
    Note that $\s_\text{opt}$ covers all clients, and is thus a social optimum profile, while $\s$ covers exactly half of the clients.
    To prove the tightness of the bound, it suffices to find a \fullEq{} $\sigma$ such that $(\s, \sigma)$ is a SPE.

    For $\sigma(\s)$ we let each client put her weight on the facility on her own vertex.
    In a facility placement profile $\s'$ where any facility $f$ deviates to any vertex $c \in V^\text{core}$, we set $\sigma(\s') = \sigma(\s)$, so that all facilities receive the same weight from the exact same client.
    For a deviation to a vertex $o \in V^\text{out}$, we set $\sigma(\s')_{o, f} = 1$ and let the client on the previous location of $f$ put her weight on an arbitrary facility other than $f$.
    We keep all other weights identical to $\sigma(\s)$ and pick arbitrary client equilibria for all yet undefined FPPs.
    In both cases $f$ does not increase her load, so profitable deviations from $\s$ are impossible.
    \fi%
\end{proof}

\ifdefined\arxiv
Additionally, \citet{ijcai-21} give a lower bound on the price of stability for a very similar model to ours.
See \Cref{sec:adaptation} for how to adapt their proof for our model.

\begin{restatable}[Price of Stability]{theorem}{pos}
    For the weighted atomic \flg{}, the price of stability regarding the utilitarian social welfare of the facility agents is at least $2-\frac1k$.
\end{restatable}
\fi

\section{Conclusion}

Even though two-stage facility location games with atomic clients may have qualitatively distinct client equilibria, we demonstrate that SPE exist for the case of unweighted clients, and that this no longer holds for non-uniform clients.
Our proof of existence via modified best response dynamics introduces several novel ideas that may be useful also beyond facility location.
For the weighted case, the conjecture that $\phi$-approximate SPE always exist remains an open problem.
%For this, we present a novel approach based on a modified best response dynamic that takes the behaviors of both clients and facilities into account. Our potential function argument crucially relies on a hierarchical classification of the clients and careful rounding in each step.
We believe our work serves as a starting point for further research into more complex competitive facility location models.

%with two or more stages.

\section*{Acknowledgements}
This work was supported by the Federal Ministry of Education and Research (BMBF) with a fellowship within the IFI program of the German Academic Exchange Service (DAAD).

\ifdefined\arxiv
    \printbibliography
\else
    \bibliographystyle{named}
    \bibliography{references}
\fi

\ifdefined\arxiv

\clearpage
\begin{center}
\noindent\textbf{\LARGE Appendix}
\end{center}
\appendix

\section{Transferring Results from the \flg{} by \texorpdfstring{\citet{ijcai-21}}{Krogmann, Lenzner, Molitor, et al. (2021)}}
\label{sec:adaptation}

\classsetalg*
\begin{proof}
We show that the algorithm to compute a MNS as defined by \citet{ijcai-21} may be used to compute an MNS for our definition as well.
To compute the minimum neighborhood set for an instance with the host graph $H$, the set of facilities $F$ and the facility placement profile $\s$ according to our definition with the given sets $F^*$ and $V^*$, we create an instance with the host graph $H'$ being the subgraph of $H$ induced by the vertices $V^*$ and use the set of facilities $F'=F^*$.
Then, we use Algorithm~2 of \citet{ijcai-21}.
Even though their definition of MNS does not demand that the MNS is the largest possible subset of $F^*$ with the given ratio, their algorithm produces this property anyway.
To see this, assume towards contradiction that there is a facility $f \notin R$ which is in $M=\MNS_\s(F^*, V^*)$.
Then the algorithm found an augmenting path in Line 15 and therefore, the facilities in $M$ can receive a total load higher than $w(A_\s(M) \cap V^*)$, which is a contradiction.

By \Cref{cor:class-set-unique} we get the unique class set through repeated computation of minimum neighborhood sets.
\end{proof}

\pos*
\begin{proof}
We reuse the construction for the graph $G$ by \citet{ijcai-21} and then add $x-1$ nodes to the star around $v_k$, such that in total the star has $k(x+1)$ nodes (including the center).
This ensures that all facilities choosing $v_k$ is the only SPE.
Then we get
\[
\text{PoS} \geq \frac{k(x+1)+(k-1)x}{k(x+1)} = \frac{2kx + k - x}{kx+k}
\]
with
\[
\lim_{x\to\infty}{\left(\frac{2kx + k - x}{kx+k}\right)}=\frac{2k-1}{k}=2-\frac1k\text.\hfill\qedhere
\]
\end{proof}

\section{SPE Counterexamples}
\nospe*
\begin{proof}
    Let $(H, 2)$ be an unrestricted instance of \flg{} with the host graph $H$ shown on the left in \Cref{fig:no-spe}.
    To show that this instance has no subgame perfect equilibrium, we first find the set of client equilibria for each induced client game.
    An exhaustive list of all client equilibria is given in \Cref{tab:equilibria}.
    \begin{table}[h]
        \centering
        \caption{All client equilibria for each facility placement profile of the instance on the left in \Cref{fig:no-spe}.} \label{tab:equilibria}
        \begin{tabular}{lll}\toprule
            $\s$ & possible client equilibria & facility loads \\ \midrule
            (1, 1) & $\forall \gamma \in [0,1]:$ & \\
                   & $\quad \big((\gamma,1-\gamma),(0,0),(0,0)\big)$& $(\gamma, 3-\gamma)$\\
            (1, 2) & \big((1,0), (0,1), (0,0)\big) & (3, 2)\\
            (1, 3) & \big((1,0), (0,0), (0,1)\big) & (3, 1)\\
            (2, 1) & \big((0,1), (1,0), (0,0)\big) & (2, 3)\\
            (2, 2) & \big((0,1), (1,0), (0,0)\big), & (2, 3)\\
            & \big((1,0), (0,1), (0,0)\big), & (3, 2)\\
            & \big((0.5,0.5), (0.5,0.5), (0,0)\big) & (2.5, 2.5) \\
            (2, 3) & \big((0,1), (1,0), (0,1)\big) & (2, 4)\\
            (3, 1) & \big((0,1), (0,0), (1,0)\big) & (1, 3)\\
            (3, 2) & \big((1,0), (0,1), (1,0)\big) & (4, 2)\\
            (3, 3) & \big((0,1), (0,0), (1,0)\big), & (1, 3)\\
            & \big((1,0), (0,0), (0,1)\big), & (3, 1)\\
            & \big((0.5,0.5), (0,0), (0.5,0.5)\big) & (2, 2) \\ \bottomrule
        \end{tabular}
    \end{table}
    We consider the facility game induced by some arbitrary $\sigma \in \Sigma$ for which we give the payoff matrix in \Cref{tab:payoff-matrix} in which the diagonal entries depend on $\sigma$.
    For each of the six off-diagonal strategy profiles, some agent can improve by moving to another off-diagonal strategy profile.
    
    \begin{table}[h]
    \centering
    \caption{The payoff matrix for the strategies of the facility game played on the instance on the left in \Cref{fig:no-spe}.} \label{tab:payoff-matrix}
    \begin{game}{3}{3}
          & $w_1$  & $w_2$ & $w_3$\\
    $w_1$ & $a_{1,1}, b_{1,1}$ & $3, 2$ & $3, 1$ \\
    $w_2$ & $2, 3$ & $a_{2,2}, b_{2,2}$ & $2, 4$ \\
    $w_3$ & $1, 3$ & $4, 2$ & $a_{3,3}, b_{3,3}$ \\
    \end{game}
    \end{table}
    
    Next, consider the diagonal strategy profiles:
    \begin{itemize}
        \item For $(w_1, w_1)$ to be stable, we require $a_{1,1} \geq 2$ and $ b_{1,1} \geq 2$ which contradicts $a_{1,1} + b_{1,1} = 3$.
        \item For $(w_2, w_2)$ to be stable, we require $a_{2,2} \geq 4$ and $b_{2,2} \geq 4$ which cannot be satisfied.
        \item For $(w_3, w_3)$ to be stable, we require: $a_{3,3} \geq 3$ and $b_{2,2} \geq 3$ which cannot be satisfied.
    \end{itemize}
    Thus, there exists an instance that does not admit a subgame perfect equilibrium.
\end{proof}

\noapproxspe*
\begin{proof}
    Let $(H, 2)$ be an unrestricted instance of \flg{} with the host graph $H$ shown on the right in \Cref{fig:no-spe}.
    We first establish that the \oneEq{} is unique if $s_f \neq s_g$, for the two facilities $f$ and $g$.
    For any such facility placement profile, the number of clients in the attraction range of both facilities is at most one.
    Thus, there is at most one client with a choice.
    We define the \emph{reach} $\rho(v)$ of a vertex $v$ to be the sum of the weights of all clients in the attraction range of a facility placed on $v$.
    Since the reach is distinct for every facility location, the unique \oneEq{} by \Cref{obs:excluded-load} is the \oneProfile{} where the client in the intersection (if there is one) exclusively considers the facility with smaller reach.
    \Cref{tab:vertex-reach} gives the reach of all vertices.
\begin{table}[h]
    \centering
    \caption{Reach (as defined in the proof of \Cref{thm:no-approx-spe}) of each vertex in the instance on the right in \Cref{fig:no-spe}.}
    \begin{tabular}{lcc}\toprule
        \textbf{vertex} & \textbf{exact reach} & \textbf{rounded reach} \\ \midrule
        $v_1$ & 2 & 2\\
        $v_2$ & $2 - \epsilon$ & $2 - \epsilon$ \\
        $v_3$ & $\phi$ & 1.618 \\
        $v_4$ & $\frac{\phi^2}{2}$ & 1.309 \\
        $v_5$ & $\frac{2}{\phi}$ & 1.236 \\
        $v_6$ & $2-\frac{2}{\phi}$ & 0.764 \\ \bottomrule
    \end{tabular}
    \label{tab:vertex-reach}
\end{table}

Let $\sigma$ denote some arbitrary\fullEq{}.
We show that for all facility placement profiles $\s$, one of the facility agents can improve their load by a factor strictly larger than $\phi - \epsilon$ by moving.
By symmetry, it suffices to show this for $\{(v_i, v_j) \in S\mid i \leq j\}$.
We first consider facility placement profiles $(s_f, s_g)$ with $s_f \neq s_g$:
\begin{itemize}
    \item For $\s \in \{(v_1, v_2), (v_1, v_3), (v_1, v_4), (v_2, v_3), (v_2, v_4),$ $(v_3, v_4)\}$, the load on $f$ is at most $2-\frac{2}{\phi} = \frac{2}{\phi^2}$.
    Facility $f$ can improve by a factor of at least $\phi$ by moving to $v_5$, since this results in a load of $\rho(v_5) = \frac{2}{\phi}$.
    \item For $\s \in \{(v_2, v_5), (v_3, v_5), (v_4,v_5)\}$, the load on facility $f$ is exactly $\frac{2}{\phi}$.
    Facility $f$ can improve by factor $\phi$ by moving to $v_1$, since this results in a load of $\rho(v_1) = 2$.
    \item For $\s = (v_1, v_5)$, the load on $g$ is $\rho(v_5) = \frac{2}{\phi}$.
    $g$ can improve by factor $\phi(1-\frac{\epsilon}{2}) > \phi - \epsilon$ by moving to $v_2$, since this results in a load of $\rho(v_2) = 2-\epsilon$.
    \item For $\s = (v_i, v_6)$, $i \in \{1, 2, 3, 4, 5\}$, the load on $g$ is $2-\frac{2}{\phi}$.
    Facility $g$ can improve by a factor of at least $\phi$ by moving to either $v_1$ or $v_5$ (depending on $s_f)$.
\end{itemize}

Next, we show that for all $v \in V$, the state $((v, v), \sigma)$ is not a $(\phi-\epsilon)$-approximate SPE.
When both facilities are located on the same vertex $v$, one of the facilities has a load of at most $\frac{\rho(v)}{2}$.
Thus, it suffices to show that for all $v \in V$, there is a vertex $u \in V$ such that the load of the moving facility $f$ is $\ell_f((u, v), \sigma) > (\phi - \epsilon) \frac{\rho(v)}{2}$.
\Cref{tab:no-symmetric-spe} lists the best response locations $u$ corresponding to each profile $\s$ with $s_f = s_g$.

\begin{table}[h]
    \centering
    \caption{Best facility response for each profile with both facilities on the same vertex in the instance on the right in \Cref{fig:no-spe}.}
    \begin{tabular}{cccc} \toprule
        vertex $v$ & $\frac{\rho(v)}{2}$ & best response $u$ & $\ell_f((u, v), \sigma)$ \\ \midrule
        $v_1$ & 1 & $v_2$ & $2-\epsilon$ \\
        $v_2$ & $1-\frac{\epsilon}{2}$ & $v_3$ & $\phi$ \\
        $v_3$ & $\frac{\phi}{2}$ & $v_4$ & $\frac{\phi^2}{2}$ \\
        $v_4$ & $\frac{\phi^2}{4}$ & $v_5$ & $\frac{2}{\phi} = 1.89 \cdot \frac{\phi^2}{4}$ \\
        $v_5$ & $\frac{1}{\phi}$ & $v_1$ & 2 \\
        $v_6$ & $1-\frac{1}{\phi}$ & $v_2$ & $2 - \epsilon$ \\ \bottomrule
    \end{tabular}
    \label{tab:no-symmetric-spe}
\end{table}
This concludes the proof.
\end{proof}

\fi

\end{document}